\documentclass[review]{elsarticle}

\usepackage{lineno,hyperref,mathtools}

\journal{Journal of \LaTeX\ Templates}

\usepackage{amssymb}
\usepackage{bm}
\usepackage{amssymb}
\usepackage{amsthm}
\usepackage{multirow,booktabs}
\usepackage{cases}
\usepackage{threeparttable}
\newtheorem{theorem}{Theorem}

\newtheorem{remark}{Remark}

\newtheorem{lemma}[theorem]{Lemma}
\newtheorem{corollary}[theorem]{Corollary}
\newtheorem{example}[theorem]{Example}
\begin{document}

\begin{frontmatter}

\title{Construction of MDS self-dual codes from generalized Reed-Solomon codes}
\tnotetext[mytitlenote]{}

\author[mymainaddress]{Ruhao Wan}
\ead{wanruhao98@163.com}

\author[mymainaddress]{Shixin Zhu\corref{mycorrespondingauthor}}
\cortext[mycorrespondingauthor]{Corresponding author}
\ead{zhushixinmath@hfut.edu.cn}

\author[mymainaddress]{Jin Li}
\ead{lijin\_0102@126.com}

\address[mymainaddress]{School of Mathematics, HeFei University of Technology, Hefei 23009, China}

\begin{abstract}
MDS codes and self-dual codes are important families of classical codes in coding theory. It is of interest to investigate MDS self-dual codes. The existence of MDS self-dual codes over finite field $F_q$ is completely solved for $q$ is even. In this paper, for finite field with odd characteristic, we construct some new classes of MDS self-dual codes by (extended) generalized Reed-Solomon codes.
\end{abstract}

\begin{keyword}
MDS self-dual codes\sep GRS code\sep Extended GRS code

\end{keyword}

\end{frontmatter}

\section{Introduction}
\label{intro}
Let $q$ be a prime power and let $F_{q}$ be the finite field with $q$ elements. An $[n,k,d]$-linear code $C$ over $F_{q}$ is a subspace of $F_{q}^{n}$ with dimension $k$ and minimum Hamming distance $d$. The Singleton bound says that $d\leq n-k+1$. If $d=n-k+1$, then $C$ is called a maximum distance separable (MDS) code. The Euclidean inner product of two vectors $\mathbf{x}=(x_{1},x_{2},...,x_{n})$ and $\mathbf{y}=(y_{1},y_{2},...,y_{n})$ in $F_{q}^{n}$ is defined by
$\mathbf{x}\cdot\mathbf{y}=\sum_{i=1}^{n} x_{i}y_{i}$.
The Euclidean dual code $C^{\bot}$ of $C$ is defined by
\begin{center}
$C^{\bot}=\{\mathbf{x}\in  F_{q}^{n}:\mathbf{x}\cdot\mathbf{c}=0,\ \forall\ \mathbf{c}\in C\}.$
\end{center}
The code $C$ is called Euclidean self-dual if $C=C^{\bot}$. If $C$ is MDS and Euclidean self-dual, $C$ is called an MDS Euclidean self-dual code.
Parameters of such codes are completely determined by its length, i.e.,$\ [n,\frac{n}{2},\frac{n}{2}+1]$ where $n$ is an even integer.
The MDS conjecture states that a non-trivial $q$-ary MDS code with parameters $[n,k,n-k+1]$ always has length $n\leq q+1$ unless $q=2^{m},n=q+2$ and $k\in\{3,q-1\}$, where $m$ is a positive integer. On the one hand, MDS codes have good properties and have attracted a lot of attention (see \cite{RefJ2,RefJ8}). A special case of MDS codes are referred to as generalized Reed-Solomon (GRS) codes. On the other hand, self-dual codes have attracted attention from coding theory, cryptography and have been found various applications in secret sharing schemes (see \cite{RefJ1,RefJ5}). More recently, the application of MDS codes renewed the interest in the construction of MDS self-dual codes (see \cite{RefJ11,RefJ12,RefJ17}).

Because parameters of MDS self-dual codes are completely determined by its length.
Then the key problem of this topic is to determine the existence of MDS self-dual codes for various lengths. This problem is completely solved for $q$ is even (see \cite{RefJ3}). It only need to consider the case $q$ is odd. Many researchers are interested in constructing MDS self-dual codes utilizing GRS codes (see \cite{RefJ15,RefJ16}).
In \cite{RefJ7}, Jin and Xing introduced a criterion to construct MDS self-dual codes through GRS codes and obtained some new MDS self-dual codes.
They also show that for any given even length $n$, there is a $q$-ary MDS code as long as $q\equiv 1({\rm mod}~4)$ and $q$ is sufficiently large.
In \cite{RefJ10}, Yan generalized the technique which was used in \cite{RefJ7} and constructed several new classes of MDS self-dual codes via GRS codes and extended GRS codes.
In \cite{RefJ9}, Xie and Fang constructed new classes of $r^m$-ary MDS self-dual codes of length $n=(t+1)r^e+1$, where $r$ is odd prime power, $t$ is even, $t\mid (r-1)$ and $\eta(t)=\eta(-1)=1$.

\newcommand{\tabincell}[2]{\begin{tabular}{@{}#1@{}}#2\end{tabular}}
\begin{table}
\caption{Parameters of some known MDS self-dual codes of length $n$.}
\label{tab:3}
\begin{center}
\resizebox{\textwidth}{85mm}{
	\begin{tabular}{cccc}
		\hline
		Class & $q$ & $n$ even & References\\
		\hline
		1 & $q$ even  &  $n \leq q$& \cite{RefJ3}  \\
		2 & $q$ odd  & $n=q+1$& \cite{RefJ3}  \\
		3 & $q$ odd & $(n-1)\mid (q-1)$, $\eta(1-n)=1$ & \cite{RefJ10}  \\
		4 & $q$ odd  & $(n-2)\mid (q-1)$, $\eta(2-n)=1$ & \cite{RefJ10}  \\
		5 & $q=r^s$, $r$ odd, $s\geq 2$  & $n=lr$, $l$ even and $2l\mid (r-1)$ & \cite{RefJ10}  \\
		6 & $q=r^s$, $r$ odd, $s\geq 2$  & $n=lr$, $l$ even, $(l-1)\mid (r-1)$ and $\eta(1-l)=1$& \cite{RefJ10}  \\
		7 & $q=r^s$, $r$ odd, $s\geq 2$  & $n=lr+1$, $l$ odd, $l\mid (r-1)$ and $\eta(l)=1$& \cite{RefJ10}  \\
		8 & $q=r^s$, $r$ odd, $s\geq 2$  & \tabincell{c}{$n=lr+1$, $l$ odd, $(l-1)\mid (r-1)$\\ and $\eta(l-1)=\eta(-1)=1$} & \cite{RefJ10}  \\
		9 & $q=r^2$, $r$ odd  & $n=tr$, $t$ even and $1\leq t\leq r$ & \cite{RefJ10}\\
	    10 & $q=r^2$, $r$ odd  & $n=tr+1$, $t$ odd and $1\leq t\leq r$ & \cite{RefJ10}\\
		11 & $q=p^k$, odd prime $p$ & $n=p^r+1$, $r\mid k$ & \cite{RefJ10}\\
		12 & $q=p^k$, odd prime $p$ & $n=2p^r$, $1\leq e <k$, $\eta(-1)=1$ & \cite{RefJ10}\\
		13 & $q\equiv1($mod$\,4)$ & $n\mid(q-1)$, $n<q-1$ & \cite{RefJ10}\\
		14 & $q\equiv1($mod$\,4)$ & $4^n\cdot n^2\leq q$ & \cite{RefJ7}\\
		15 & $q=r^2$ & $n\leq r$ & \cite{RefJ7} \\
		16 & $q=r^2$, $r\equiv 3($mod$\,4)$  & $n=2tr$ for any $t\leq \frac{r-1}{2}$& \cite{RefJ7}\\
        17 & $q=p^m\equiv 1($mod$\,4)$& $n=p^l+1$ with $0\leq l\leq m$ & \cite{RefJ4}\\
        18 & $q=r^m$, $m$ even, odd $r$& $n=tr^l$ with even $t$, $1\leq t\leq r-1$, $1\leq l<m$ & \cite{RefJ4}\\
        19 & $q=r^m$, $m$ even, odd $r$& $n=tr^l+1$ with odd $t$, $1\leq t \leq r-1$, $1\leq l<m$ & \cite{RefJ4}\\
        20 & $q=p^m$, $p$ odd prime & $n=2tp^e$, $2t\mid (q-1)$ even and $e<m$ & \cite{RefJ14}\\
        21 & $q=r^2$, $r$ odd& $n=tm$, $1\leq t \leq \frac{r-1}{gcd(r-1,m)}$, $\frac{q-1}{m}$ even & \cite{RefJ14}\\
        22 & $q=r^2$, $r$ odd& \tabincell{c}{$n=tm+1$, $tm$ odd, $1\leq t \leq \frac{r-1}{gcd(r-1,m)}$\\ and $m\mid (q-1)$} & \cite{RefJ14}\\
        23 & $q=r^2$, $r$ odd& \tabincell{c}{$n=tm+2$, $tm$ even, $1\leq t \leq \frac{r-1}{gcd(r-1,m)}$\\ and $m\mid (q-1)$} & \cite{RefJ14}\\
        24 & $q=r^2$, $r\equiv 1($mod$\,4)$ and $s$ even & \tabincell{c}{$n=s(r-1)+t(r+1)$ with $1\leq s\leq \frac{r+1}{2}$\\ and $1\leq t\leq \frac{r-1}{2}$} & \cite{RefJ18}\\
        25 & $q=r^2$, $r\equiv 3($mod$\,4)$ and $s$ odd & \tabincell{c}{$n=s(r-1)+t(r+1)$ with $1\leq s\leq \frac{r+1}{2}$\\ and $1\leq t\leq \frac{r-1}{2}$} & \cite{RefJ18}\\
        26 & $q=r^s$, $r$ odd and $s\geq 2$ &  \tabincell{c}{$n=(t+1)r^z$, odd $t$ with $t\mid (r-1)$, \\$\eta(-t)=1$ and $1\leq z\leq s-1$ }& \cite{RefJ6},\cite{RefJ9}\\
        27 & $q=r^s$, $r$ odd and $s\geq 2$ &  \tabincell{c}{$n=tr^z+1$, odd $t$ with $t\mid (r-1)$, \\$\eta((-1)^{\frac{r^z+1}{2}}t)=1$ and $1\leq z\leq s-1$ }& \cite{RefJ6}\\
        28 & $q=r^s$, $r$ odd and $s\geq 2$ &  \tabincell{c}{$n=(t+1)r^z+1$, even $t$ with $t\mid (r-1)$, \\$\eta(-1)=\eta(t)=1$ and $1\leq z\leq s-1$ }& \cite{RefJ9}\\
        29 & $q=p^s\equiv 3($mod$\,4)$& $n=p^e+1$ odd $e$ with $1\leq e\leq s$ & \cite{RefJ6}\\
		\hline
	\end{tabular}}
    \begin{tablenotes}
     \footnotesize
    \item Note: In \cite{RefJ6}, \cite{RefJ15}, \cite{RefJ16} and \cite{RefJ20} many MDS codes on $q=r^2$ are also introduced.
    \end{tablenotes}
\end{center}
\end{table}

In \cite{RefJ13}, Zhang and Feng showed that when $q\equiv 3($mod$\,4)$ and $n\equiv 2($mod$\,4)$, $q$-ary MDS self-dual codes of length $n$ do not exist.
Afterwards, in \cite{RefJ14}, Lebed and Liu showed that there exist $p^m$-ary MDS self-dual codes of length $n=2tp^e$, for any $t$ with $2t\mid(p-1)$, $e<m$ and $\frac{q-1}{2t}$ is even.
In \cite{RefJ6}, they also showed that there exist MDS self-dual codes of length $n=tr^z+1$, where $t\mid r-1$ and $\eta((-1)^{\frac{r^z+1}{2}}t)=1$, for any  integer $z$ with $1\leq z\leq s-1$.

In the paper, we obtain some new results on the existence of MDS self-dual codes by (extended) GRS codes. Precisely, our main contribution is to construct new MDS self-dual codes (see Table \ref{tab:3}).

\begin{table}
\caption{Our results}
\label{tab:3}
\begin{center}
\resizebox{\textwidth}{50mm}{
	\begin{tabular}{ccc}
		\hline
		$q$  & $n$ even & References\\
		\hline
		$q=r^m$, $r$ odd and $q\equiv 1({\rm mod}~4)$   &  $n=2tr^e$, $0\leq e\leq m-1$, $2t\mid r-1$ and $t\neq \frac{r-1}{2}$ & Theorem \ref{th1} \\
		$q=r^m$, $r$ odd                                &  \tabincell{c}{$n=(t+1)r^e+1$, $0\leq e\leq m-1$, $t$ even and $t\mid r-1$,\\
                                                                      if $\eta(t)=\eta(-1)=1$, or $\eta(-t)=1$ and $e$ even}    & Theorem \ref{th4} \\
		$q=p^m$, odd prime $p$ and $q\equiv 1({\rm mod}~4)$ & \tabincell{c}{$n=(t+1)p^e$, $0\leq e\leq m-1$, $2\leq t\leq p-1$ and $t$ odd,\\
                                                                      if $\eta(i(t+1-i))=1$ for $1\leq i\leq \frac{t-1}2$}      & Theorem \ref{th2} \\
		$q=p^m$, odd prime $p$ and $q\equiv 1({\rm mod}~4)$ & \tabincell{c}{$n=(t+1)p^e+1$, $0\leq e\leq m-1$, $2\leq t\leq p-1$ and $t$ even,\\
                                                                      if $\eta(i (t+1-i))=1$ for $1\leq i\leq \frac{t}2$}      & Theorem \ref{th3} \\
		$q=r^{sm}$, $r$, $m$ odd and $q\equiv 1({\rm mod}~4)$ & \tabincell{c}{$n=tr^e(1+r^s+\dots+r^{s(m-1)})$, $0\leq e \leq s-1$, $0<t<r-1$,\\
                                                                    if $t$ even and $t\mid r-1$}                           & Theorem \ref{th8} \\
		$q=r^{sm}$, $r$, $m$ odd                 & \tabincell{c}{$n=(t+1)r^e(1+r^s+\dots+r^{s(m-1)})$, $0\leq e \leq s-1$, $0<t<r$,\\
                                                                       if $t$ odd, $t\mid r-1$ and $\eta(-t)=1$}            & Theorem \ref{th9} \\
		$q=r^{sm}$, $r$, $m$ odd                 & \tabincell{c}{$n=tr^e(1+r^s+\dots+r^{s(m-1)})+1$, $0\leq e\leq s-1$, $0<t<r$,\\
                                                                       if $t$ odd, $t\mid r-1$ and $\eta((-1)^{\frac{r^e+1}{2}}t)=1$} & Theorem \ref{th10} \\
		$q=r^{sm}$, $r$, $m$ odd                 & \tabincell{c}{$n=(t+1)r^e(1+r^s+\dots+r^{s(m-1)})+1$, $0 \leq e\leq s-1$, $0<t<r-1$,\\
                                                        $t$ even and $t\mid r-1$, if $\eta(t)=\eta(-1)=1$, or $\eta(-t)=1$ and $e$ even } & Theorem \ref{th11}\\
		$q=r^2$, $r$ odd and $q-1=ef$           & \tabincell{c}{$n=tf$, $1\leq t\leq \frac{s(r+1)}{gcd(s(r+1),f)}$, $s\mid f$ and $s\mid r-1$,\\
                                                if $e$ and $\frac{r-1+tf}{s}$ even} & Theorem \ref{th tf tf+2}(1) \\
		$q=r^2$, $r$ odd and $q-1=ef$           & \tabincell{c}{$n=tf+2$, $1\leq t\leq \frac{s(r+1)}{gcd(s(r+1),f)}-1$, $s\mid f$ and $s\mid r-1$,\\
                                                if $\frac{f}{s}$ and $\frac{1}{2}(t-1)(r+1)$ even, or $\frac{f}{s}$ odd and $t$ even} & Theorem \ref{th tf tf+2}(2) \\
		$q=r^2$, $r$ odd and $q-1=ef$           & \tabincell{c}{$n=tf+2$, $t=\frac{s(r+1)}{gcd(s(r+1),f)}$, $s\mid f$ and $s\mid r-1$,\\
                                                if $\frac{ft}{s}$ and $\frac{t-1}{2}(r+1-\frac{ft}{s})$ even} & Theorem \ref{th tf tf+2}(3)\\
		$q=r^2$, $r$ odd and $q-1=ef$           & \tabincell{c}{$n=tf+1$, $1\leq t\leq \frac{s(r-1)}{gcd(s(r-1),f)}$, $s\mid f$ and $s\mid r+1$,\\
                                                if $tf$ odd} & Theorem \ref{th tf tf+1 tf+2} \\
		$q\equiv 1({\rm mod}~4)$              & $q>(t+(t^2+(n-1)2^{n-2})^{\frac{1}{2}})^2$, where $t=(n-3)2^{n-3}+\frac{1}{2}$ & Theorem \ref{th lager q} \\
		\hline
	\end{tabular}}
\end{center}
\end{table}
This paper is organized as follows.
In Section 2, we will introduce some basis knowledge and auxiliary results on GRS codes and extended GRS codes.
In Section 3, by (extended) GRS codes, we present some new MDS self-dual codes.
Finally, we give a short summary of this paper in Section 4.

\section{Preliminaries}

In this section, we recall some basic properties and results about GRS and extended GRS codes.
Throughout this paper, let $F_{q}$ be the finite field with $q$ elements, and let $n$ be a positive integer with $1<n<q$.

A multiplicative character of $F_q$ is a nonzero function $\psi$ from $F_q^*$ to the set of complex numbers such that $\psi(xy)=\psi(x)\psi(y)$
for all $x,y\in F_q^*$.
Let $\theta$ be a fixed primitive element of $F_q^*$.
For each $j=0,1,\dots,q-2$, the function with
\begin{equation}
\psi_j(\theta^k)=e^{\frac{2ijk\pi}{q-1}},\quad for\quad k=0,1,\dots,q-1
\end{equation}
defines a multiplicative character.
The multiplicative character $\eta=\psi_{\frac{q-1}{2}}$ is called quadratic character of $F_q$, i.e.,
$\eta(\theta^k)=1$ if and only if $k$ is even.

Choose $\mathbf{a}=\{\alpha_{1},\alpha_{2},...,\alpha_{n}\}$ to be an $n$-tuple of distinct elements of $F_{q}$,
and let $\mathbf{v}=\{v_{1},v_{2},...,v_{n}\}$ with $v_{i}\in F_{q}^{*}$.
For an integer $k$ with $0\leq k\leq n$, the generalized Reed-Solomon (GRS) code is defined by
\begin{equation}\label{eq1}
GRS_{k}(\mathbf{a},\mathbf{v})=\{(v_{1}f(\alpha_{1})),v_{2}f(\alpha_{1}),\dots,v_{n}f(\alpha_{n}):f(x)\in F_{q}[x],\deg(f(x))\leq k-1\}.
\end{equation}
It is well known that the code $GRS_{k}(\mathbf{a},\mathbf{v})$ is a $q$-ary $[n,k,n-k+1]$ MDS codes, and its dual is also MDS. Let
\begin{equation}\label{EQ3}
f_{\mathbf{a}}(x)=\prod_{\alpha\in \mathbf{a}}(x-\alpha)\quad and \quad L_{\mathbf{a}}(\alpha_i)=\prod_{1\leq j\leq n,j\neq i}(\alpha_i-\alpha_j),
\end{equation}
which will be used frequently in this paper.

The following lemma gives a criterion for a GRS code to be self-dual.
\begin{lemma}(\cite{RefJ7})\label{lem1}
Assume the notations given above.
For an even integer $n$, and $k=\frac{n}{2}$,
if there exists an element $\lambda\in F_q^*$ such that $\eta(\lambda L_\mathbf{a}(a_i))=1$ for all $1\leq i \leq n$,
then the code $GRS_k(\mathbf{a},\mathbf{v})$ defined in (\ref{eq1}) is an MDS self-dual code, where $v_i^2=(\lambda L_\mathbf{a}(a_i))^{-1}$ for $1\leq i \leq n$.
\end{lemma}
The above lemma can be rewritten as follow.
\begin{remark}(\cite{RefJ13})\label{rem1}
For an even integer $n$, and $k=\frac{n}{2}$, if $\eta(L_{\mathbf{a}}(a_{i}))$ are the same for all $1\leq i \leq n$,
then the code $GRS_k(\mathbf{a},\mathbf{v})$ defined in (\ref{eq1}) is an MDS self-dual code.
\end{remark}

Now we introduce some basic notations and results on extended GRS codes.
The $k$-dimensional extended GRS codes of length $n$ is defined by
\begin{equation}\label{eq2}
GRS_{k}(\mathbf{a},\mathbf{v},\infty)=\{(v_{1}f(\alpha_{1})),\dots,v_{n-1}f(\alpha_{n-1}),f_{k-1}:f(x)\in F_{q}[x],\deg(f(x))\leq k-1\},
\end{equation}
where $f_{k-1}$ is the coefficient of $x^{k-1}$ in $f(x)$. It is well known that the code $GRS_{k}(\mathbf{a},\mathbf{v},\infty)$ is a $q$-ary $[n+1,k,n-k+1]$ MDS codes, and its dual is also MDS.

\begin{lemma}(\cite{RefJ10})\label{lem2}
For an odd integer $n$, and $k=\frac{n+1}{2}$, if $\eta(-L_{\mathbf{a}}(a_i))=1$ for all $1\leq i \leq n$,
then the code $GRS_k(\mathbf{a},\mathbf{v},\infty)$ defined in (\ref{eq2}) is an MDS self-dual code, where $v_i^2=(-L_\mathbf{a}(a_i))^{-1}$ for all $1\leq i \leq n$.
\end{lemma}

We give the following lemmas, which are useful in the proof of the main results.
\begin{lemma}(\cite{RefJ6})\label{lem3}
Suppose $q=r^m$, where $r$ is an odd prime power and $m$ is a positive integer.
Let $V$ be an $ F_r$-vector subspace of dimension $e$ in $F_q$, where $1\leq e\leq m$.
Then
\begin{equation}
\prod_{0\neq v\in V}v=(-1)^{\frac{r^e-1}{2}}\delta^2,
\end{equation}
for some $\delta\in F_q^*$.
\end{lemma}

\begin{lemma}(\cite{RefJ10})\label{lem4}
Let $m\mid (q-1)$ be a positive integer and let $\alpha\in F_q$ be a primitive $m$-th root of unity.
Then, for any $1\leq i\leq m$,
\begin{equation}
\prod_{1\leq j \leq m,j\neq i}(\alpha^i-\alpha^j)=m\alpha^{(m-1)i}=m\alpha^{-i}.
\end{equation}
\end{lemma}

\begin{lemma}(\cite{RefJ18})\label{lem5}
(1) Let $S_1$ and $S_2$ be disjoint subsets of $F_q$, $S=S_1\bigcup S_2$.
Then for $b\in S$,
\[ L_{S}(b)=\begin{cases}
L_{S_1}(b)f_{S_2}(b),\quad if\  b\in S_1 \\
L_{S_2}(b)f_{S_1}(b),\quad  if\  b\in S_2.
\end{cases}\]\\
(2) Let $\theta$ be a primitive element of $F_q^*$ and $q-1=ef$.
Denote $H=\langle\theta^{e}\rangle$, then
\begin{equation}
f_{\theta^i H}(x)=x^{f}-\theta^{if} \quad and \quad L_{\theta^i H}(x)=f'_{\theta^i H}(x)=fx^{f-1}.
\end{equation}
\end{lemma}

\begin{lemma}(\cite{RefJ13})\label{lem xin}
Let $\theta$ be a primitive element of $F_q^*$ and $q-1=ef$.
Denote $H=\langle\theta^{e}\rangle$.
Let $S_i=\xi_iH$ $(1\leq i\leq t)$ be $t$ distinct cosets of $H$ in $F_q^*$ $(0\leq t\leq e-1)$,
$S=\bigcup_{i=1}^t S_i$, $\mid S\mid=tf$.
Then
\begin{equation}
f_S(x)=\prod_{i=1}^{t}\prod_{j=0}^{f}(x-\xi_i\theta^{ej})=\prod_{i=1}^{t}(x^{f}-\xi_i^f)=g(x^f),
\end{equation}
where $g(x)=\prod_{i=1}^t(x-\xi_i^f)=f_{S'}(x)$, $S'=\{\xi_i^f: 1\leq i \leq t\}$.
\end{lemma}

\section{Main results}
In this section, we construct some new classes of MDS self-codes over finite fields with odd characteristic by GRS codes and extended GRS codes.

\begin{lemma} \label{lem7}
Let $q=r^m$, where $r$ is an odd prime power and $m$ is a positive integer.
Let $V=\{v_1,v_2,\ldots, v_{r^e} \}$ be a $e$-dimensional $F_r$-vector subspace in $F_q$, where $0\leq e \leq m-1$. For any even integer $t$ with $1\leq t \leq r-1$, let $\mathbf{b}=(\beta_1,\beta_2,\dots,\beta_t)$ be a $t$-tuple of distinct elements in $F_r$.
Let $\zeta$ be a fixed element in $F_q\setminus V$.
Let
\begin{equation}
W=\bigcup_{j=1}^t(\beta_j\zeta+V).
\end{equation}
Let $\mathbf{a}=(\alpha_{11},\dots,\alpha_{1r^e},\dots,\alpha_{tr^e})$, where $\alpha_{ij}=\beta_i\zeta+v_j$. Then there exits a $\lambda\in  F_q^*$ such that $\eta(\lambda L_{\mathbf{a}}(\alpha_{ij}))=\eta(L_{\mathbf{b}}(\beta_i))$ for any $1\leq i \leq t$, $1\leq j \leq r^e$.
\end{lemma}

\begin{proof}
For $b=\beta_k\zeta+v$ with $1\leq k\leq t$ and $v\in V$,
\[\begin{split}
	L_{\mathbf{a}}(b)&= \prod_{1\leq i\leq t,1\leq j \leq r^e,b \neq \alpha_{ij}}(b-\alpha_{ij})\\
	&= \prod_{v\neq v_j} (v-v_{j}) \cdot \prod_{1\leq i\leq t,i\neq k,1\leq j \leq r^e, }[(\beta_k-\beta_i)\zeta+v-v_j] \\
	&=(\prod_{0\neq v\in V}v)\cdot \prod_{i=1,i\neq k}^t [(\beta_k-\beta_i)^{r^e} \prod_{j=1}^{r^e}(\zeta+(\beta_k-\beta_i)^{-1}(v-v_j))]\\
	&=(\prod_{0\neq v\in V}v)\cdot \prod_{i=1,i\neq k}^t [(\beta_k-\beta_i) \prod_{v\in V}(\zeta+v)]\\
	&=(\prod_{0\neq v\in V}v)\cdot \prod_{v\in V}(\zeta+v)^{t-1} \cdot \prod_{i=1,i\neq k} (\beta_k-\beta_i) \\
	&=(\prod_{0\neq v\in V}v)\cdot \prod_{v\in V}(\zeta+v)^{t-1} \cdot L_{\mathbf{b}}(\beta_k). \\
	\end{split}\]
 By Lemma \ref{lem3}, there is $\delta \in F_q^*$ such that $\prod_{0\neq v\in V}v=(-1)^{\frac{r^e-1}{2}}\delta^2$. Let
  $$\lambda=(-1)^{\frac{r^e-1}{2}}\delta^2\cdot \prod_{v'\in V}(\zeta+v')^{t-1},$$
  then $\lambda \in F_q^*$ and $L_{\mathbf{a}}(b)=\lambda\cdot L_{\mathbf{b}}(\beta_k)$. It follows that $\eta(\lambda L_{\mathbf{a}}(b))=\eta(L_{\mathbf{b}}(\beta_k))$. This completes the proof.
\end{proof}

\begin{theorem}\label{th1}
Let $q=r^m$ and $q\equiv 1 ({\rm mod}~4)$, where $r$ is an odd prime power and $m$ is a positive integer. Let $e$ be an integer with $0\leq e\leq m-1$, and let $t$ be a positive divisor of $\frac{r-1}2$, and $t\neq \frac{r-1}2$. Then there exists a $q$-ary MDS self-dual code of length $2t r^e$.
\end{theorem}

\begin{proof}
The proof can be divided into the following two cases.

{\bf Case 1}: $t$ is odd. Then $2t \equiv 2({\rm mod}~4)$. Let $\beta\in F_r^*$ be a primitive $2t$-th root of unity and $\mathbf{b}=\{\beta^i:1\leq i\leq 2t\}$. By Lemma \ref{lem4}, for any $1\leq i\leq 2t$,
\begin{equation}
L_{\mathbf{b}}(\beta^i)=\prod_{1\leq j \leq n,j\neq i}(\beta^i-\beta^j)=2t\beta^{-i}.
\end{equation}
Since $2t\mid (r-1)$ and $r-1\mid (q-1)$, this is a primitive element $\theta$ of $F_q$ such that $\beta=\theta^{\frac{q-1}{2t}}$. Since $q\equiv 1({\rm mod}~4)$ and $2t\equiv 2({\rm mod}~4)$, we have $\frac{q-1}{2t}$ is even. It follows that $\eta(\beta)=1$. By Lemma \ref{lem7}, there is $\lambda_1\in F_q^*$ such that $$\eta(\lambda_1 L_{\mathbf{a}}(\alpha_{ij}))=\eta(L_{\mathbf{b}}(\beta^i)).$$
Let $\lambda=2t\lambda_1$, then
\[\begin{split}
	\eta(\lambda L_{\mathbf{a}}(\alpha_{ij}))&=\eta(2t\cdot L_{\mathbf{b}}(\beta^i))\\
	&=\eta((2t)^2 \beta^{-i})\\
	&=\eta(\beta^{-i})=1.
\end{split}\]
By Lemma \ref{lem1}, the code $GRS_k(\mathbf{a},\mathbf{v})$ defined in (\ref{eq1}) is an MDS self-dual code. The desired result follows.

{\bf Case 2}: $t$ is even. Then $2t\equiv 0({\rm mod}~4)$. Let $\beta\in F_r^*$ be a primitive $t$-th root of unity. It is easy to check that $\eta(\beta)=1$. Let $E_r$ denote the set of nonzero squares of $ F_r^*$. Note that $t<\frac{r-1}2$, there is $\zeta\in{E_r}\backslash \{\beta^i:1\leq i\leq t\}$. Let
$$\mathbf{b}=(\beta,\beta^2,\dots,\beta^t,\zeta\beta,\zeta\beta^2,\dots,\zeta\beta^t).$$
For any $\delta \in F_r$, it is clear that $x^t-\delta^t=\prod_{1\leq j\leq t}(x-\delta\beta^j)$.
Then
\begin{equation}
L_\mathbf{b}(\beta^i)=\prod_{1\leq j \leq t, j\neq i}(\beta^i-\beta^j) \prod_{1\leq j\leq t}(\beta^i-\zeta \beta^j)=\beta^{-i}t(1-\zeta^t)
\end{equation}
and
\begin{equation}
L_\mathbf{b}(\zeta\beta^i)=\prod_{1\leq j \leq t}(\zeta\beta^i-\beta^j)\prod_{1\leq j\leq t,i\neq j}(\zeta\beta^i-\zeta\beta^j)=\zeta^{t-1}\beta^{-i}t(\zeta^t-1).
\end{equation}
By Lemma \ref{lem7}, there is $\lambda_1\in F_q^*$ such that $\eta(\lambda_1 L_{\mathbf{a}}(\alpha_{ij}))=\eta(L_{\mathbf{b}}(\beta^i))$ and $$\eta(\lambda_1 L_{\mathbf{a}}(\alpha_{t+i, j}))=\eta(L_{\mathbf{b}}(\zeta\beta^i)).$$
It is clear that $\eta(-1)=\eta(\zeta)=\eta(\beta)=1$. Let $\lambda=t(\zeta^t-1)\lambda_1$, then
\[ \begin{split}
 	\eta(\lambda L_{\mathbf{a}}(\alpha_{ij}))&=\eta(t(\zeta^t-1)\cdot L_{\mathbf{b}}(\beta^i))\\
 	&=\eta(\beta^{-i}t^2(1-\zeta^t)^2)\\
 	&=1
 \end{split}\]
 and
\[ \begin{split}
 	\eta(\lambda L_{\mathbf{a}}(\alpha_{t+i, j}))&=\eta(t(\zeta^t-1)\cdot L_{\mathbf{b}}(\zeta \beta^i))\\
 	&=\eta(\zeta^{t-1}\beta^{-i}t^2 (\zeta^t-1)^2)\\
 	&=1.
 \end{split}\]
By Lemma \ref{lem1}, the code $GRS_k(\mathbf{a},\mathbf{v})$ defined in (\ref{eq1}) is an MDS self-dual code. The desired result follows.
\end{proof}

\begin{remark}
Let $r$ be an odd prime power, and $q=r^m$. The following hold.
\begin{itemize}
\item In \cite{RefJ14}, the authors proved that, if $2t\mid (r-1)$ and $\frac{q-1}{2t}$ is even, there exists a $q$-ary MDS self-dual code of length $n=2t p^e$, where $0\leq e\leq m-1$.
\item In Theorem \ref{th1}, we proved that, if $2t\mid(r-1)$ and $q \equiv 1({\rm mod}~4)$, there exists a $q$-ary MDS self-dual code of length $n=2t r^e$, where $0\leq e\leq m-1$.
\end{itemize}
This shows that Theorem \ref{th1} extends the existence of an MDS self-dual codes.
\end{remark}

\begin{example}
Let $r=9$, $q=9^2=81$, $t=2$, we have $q\equiv 1({\rm mod}~4)$ and $2t\mid r-1$. By Theorem \ref{th1}, there exists an MDS self-dual code of length $n=2tr=36$.
\end{example}

\begin{theorem}\label{th2}
Let $q=p^m$ and $q\equiv 1({\rm mod}~4)$, where $p$ is an odd prime and $m$ is a positive integer. Let $t$ be an odd integer with $2\leq t \leq p-1$, and let $0\leq e \leq m-1$. If $\eta(i(t+1-i))=1$ for $1\leq i\leq \frac{t-1}2$. Then there exists a $q$-ary MDS self-dual code of length $(t+1)p^e$.
\end{theorem}

\begin{proof}
We choose $\mathbf{a}=\{0,1,\dots,t\}\subseteq F_p$. By Equation (\ref{EQ3}),
\begin{align*}
	L_{\mathbf{a}}(i)&=\prod_{0\leq j\leq t,  j\neq i}(i-j), \\
	&=(-1)^t\prod_{0\leq j\leq t,  j\neq i}(j-i)\\
	&=-\prod_{0\leq j\leq t,  j\neq i}(j-i),
\end{align*}
for $i\in \{0,1, \dots, t\}$. It follows that
\begin{align*}
\eta(L_{\mathbf{a}}(i))=
 \eta((-1)^{i-1}i!(t-i)!).
\end{align*}
If $0\leq i\leq \frac{t-3}2$, then
\begin{align*}
\eta(L_{\mathbf{a}}(i))&=\eta((-1)^{i-1}(i+1)(i+2)\cdots(t-i))\\	
&=[\eta(-1)]^{i-1}\prod_{j=i+1}^{\frac{t-1}2} \eta(j(t+1-j))\eta(\frac{t+1}{2}).
\end{align*}
If $ i= \frac{t-1}2$, then
\begin{align*}
\eta(L_{\mathbf{a}}(i))&=[\eta(-1)]^{i-1}\eta(\frac{t+1}{2}).
\end{align*}
Since $q\equiv 1({\rm mod}~4)$, we have $\eta(-1)=1$. By assumption, $\eta(i (t+1-i))=1$ for all $1\leq i \leq \frac{t-1}2$.
From symmetry, we have $\eta(L_{\mathbf{a}}(i))=\eta(L_{\mathbf{a}}(t-i))$, for $0\leq i\leq \frac{t-1}2$.
It follows that $\eta(L_{\mathbf{a}}(\alpha))$ are same for all $\alpha\in \mathbf{a}$. By Remark \ref{rem1} and Lemma \ref{lem7}, there exists a $q$-ary MDS self-dual code of length $(t+1)p^e$. This completes the proof.
\end{proof}

\begin{example}
Let $p=13$, $q=p^2=169$ and $t=3$,
we have $q\equiv 1({\rm mod}~4)$, $q\equiv 1({\rm mod}~8)$ and $q\equiv 1({\rm mod}~12)$,
therefore $\eta(-1)=\eta(3)=\eta(2)=1$.
By Theorem \ref{th2}, there exists a $q$-ary MDS self-dual code of length $n=(t+1)p=52$.
\end{example}

\begin{lemma}(\cite{RefJ9})\label{lem8}
Let $q=r^m$, where $r$ is an odd prime power and $m$ is a positive integer. Let $V$ be an $F_r$-subspace of $F_q$ of dimension $e$, where $0\leq e<m$. Suppose $\mathbf{a}=(\alpha_1,\alpha_2,\dots,\alpha_t)\in  F_r^t$, $t$ is odd, and $\mathbf{v}\in (F_q^*)^t$ such that $GRS_{\frac{t+1}{2}}(\mathbf{a},\mathbf{v},\infty)$ is self-dual. If $q\equiv 1({\rm mod}~4)$ or $e$ even, then there exists a $q$-ary MDS self-dual code of length $tr^e+1$.
\end{lemma}

\begin{theorem}\label{th3}
Let $q=p^m$ and $q\equiv 1({\rm mod}~4)$, where $p$ is an odd prime and $m$ is a positive integer. Let $t$ be an even integer such that $2\leq t \leq p-1$ and let $0\leq e\leq m-1$. If $\eta(i (t+1-i))=1$ for $1\leq i\leq \frac{t}2$. Then there exists a $q$-ary MDS self-dual code of length $(t+1)p^e+1$.
\end{theorem}

\begin{proof}
We choose $\mathbf{a}=\{0,1,\dots,t\}\subseteq F_p$. By Theorem \ref{th2},
\begin{align*}
\eta(L_{\mathbf{a}}(i))=
 \eta((-1)^{i}i!(t-i)!).
\end{align*}
If $0\leq i \leq \frac{t-2}2$, then
\begin{align*}
\eta(L_{\mathbf{a}}(i))&=\eta((-1)^{i}(i+1)(i+2)\cdots(t-i))\\	
&=[\eta(-1)]^{i}[\prod_{j=i+1}^{\frac{t}2} \eta(j(t+1-j))].
\end{align*}
If $i=\frac{t}2$, then
\begin{align*}
\eta(L_{\mathbf{a}}(i))&=\eta((-1)^{i})=[\eta(-1)]^{i}.
\end{align*}
Since $q\equiv 1({\rm mod}~4)$, we have $\eta(-1)=1$. By assumption, $\eta(i (t+1-i))=1$ for all $1\leq i \leq \frac{t}2$.
From symmetry, we have $\eta(L_{\mathbf{a}}(i))=\eta(L_{\mathbf{a}}(t-i))$, for $0\leq i\leq \frac{t}2$.
It follows that $\eta(-L_{\mathbf{a}}(\alpha))=1$ for all $\alpha\in \mathbf{a}$. By Lemma \ref{lem2} and Lemma \ref{lem8}, there exists a $q$-ary MDS self-dual code of length $(t+1)p^e+1$. This completes the proof.
\end{proof}

\begin{example}
Let $p=13$, $q=p^2=169$ and $t=2$,
we have $q\equiv 1({\rm mod}~4)$ and $q\equiv 1({\rm mod}~8)$,
therefore $\eta(-1)=\eta(2)=1$.
By Theorem \ref{th3},
there exists an MDS self-dual code of length $n=(t+1)p+1=40$.
\end{example}

\begin{theorem}\label{th4}
Let $q=r^m$, where $r$ is an odd prime power and $m\geq 1$. Let $t$ be an even integer and $t\mid r-1$. If $\eta(t)=\eta(-1)=1$, or $\eta(-t)=1$ and $e$ is even, then there exists a $q$-ary MDS self-dual code with length $(t+1)r^e+1$, where $0\leq e\leq m-1$.
\end{theorem}

\begin{proof}
Let $\beta\in  F_r^*$ be a primitive $t$-th root of unity and $\mathbf{b}=(0,\beta,\beta^2,\dots,\beta^t)$. By Lemma \ref{lem4}, then for $1\leq i \leq t$,
\[\begin{split}
	L_{\mathbf{b}}(\beta^i)&=\beta^i \prod_{1\leq j\leq t, j\neq i}(\beta^i-\beta^j)=\beta^{it}t=t\\
	\end{split}\]
and
\[\begin{split}
	L_{\mathbf{b}}(0)&=\prod_{1\leq j\leq t}(0-\beta^j)=-1.\\
	\end{split}\]
Since $\eta(-t)=1$, $GRS_{\frac{t+2}{2}}(\mathbf{b},\mathbf{v},\infty)$ is self-dual.
By Lemma \ref{lem8}, if $q\equiv 1({\rm mod}~4)$ or $e$ even, there exists a $q$-ary MDS self-dual code of length $(t+1)r^e+1$. This completes the proof.
\end{proof}

\begin{remark}
Let $r$ be an odd prime power, and $q=r^m$. The following hold.
\begin{itemize}
\item In \cite{RefJ9}, the authors proved that, if $t\mid (r-1)$ and $\eta(t)=\eta(-1)=1$, there exists a $q$-ary MDS self-dual code of length $n=(t+1)r^e+1$, where $0\leq e\leq m-1$.
\item In Theorem \ref{th4}, we proved that, if $t\mid (r-1)$, $\eta(-t)=1$ and $e$ is even, there exists a $q$-ary MDS self-dual code of length $n=(t+1)r^e+1$, where $0\leq e\leq m-1$.
\end{itemize}
This shows that Theorem \ref{th4} extends the existence of an MDS self-dual codes.
\end{remark}

\begin{example}
Let $r=11$, $q=r^3=1331$ and $t=2$,
we have $t\mid r-1$, and $\eta(-2)=1$.
By Theorem \ref{th4},
when $e=2$,
there exists an MDS self-dual code of length $n=(t+1)r^e+1=3\times 121+1=364$.
\end{example}

Now we consider the union of cosets from multiplicative subgroup of $F_q^*$.\\
For brevity,\\
 $\cdot\ $ Let $\theta$ be a primitive element of $F_q^*$.\\
 $\cdot\ $ $q-1=e_1e_2$.\\
 $\cdot\ $ $H_1=\langle\theta^{e_1}\rangle$, $H_2=\langle\theta^{e_2}\rangle$.\\
 $\cdot\ $ $v(\alpha)=min\{x\in N| \alpha=\theta^{xe_1}\}$ for $\alpha\in H_1$.

\begin{lemma}(\cite{RefJ9})\label{lem9}
Let $t$ be even. Suppose $GRS_{\frac{t}{2}}(\mathbf{a},\mathbf{v})$ is self-dual for some $\mathbf{a}=\{\alpha_1,\alpha_2,\dots, \alpha_t \}\subseteq H_1$ and $\mathbf{v}\in (F_q^*)^{t}$. If $e_1$ is odd, then there exists a $q$-ary MDS self-dual code of length $t e_1$.
\end{lemma}

\begin{theorem}\label{th8}
Let $q=r^{s m}$ and $q\equiv 1({\rm mod}~4)$, where $r$ is an odd prime power, $s\geq 1$ and $m$ is odd. Let $t$ be an even integer such that $t\mid (r-1)$ and $1<t<r-1$, then there exists a $q$-ary MDS self-dual code of length $tr^e(1+r^s+\dots+r^{s(m-1)})$, where $0\leq e \leq s-1$.
\end{theorem}

\begin{proof}
Let $e_1=1+r^s+\dots+r^{s(m-1)}$, then $H_1=\langle \theta^{e_1}\rangle=F_{r^s}^*$. Since $m$ is odd, we have $e_1$ is odd. Note that $r^s\equiv1({\rm mod}~4)$,
by Theorem \ref{th1},
there exists a self-dual $GRS_{\frac{tr^e}{2}}(\mathbf{a}, \mathbf{v})$, for some $\mathbf{a}=\{\alpha_1,\alpha_2,\dots,\alpha_{tr^e}\}\subseteq F_{r^s}$,
$tr^e$ is even, and $\mathbf{v}\in(F_q^*)^{tr^e}$.
For any $a\in F_{r^s}$,
let $a+\mathbf{a}=\{a+\alpha_1,a+\alpha_2,\dots,a+\alpha_{tr^e}\}$.
Obviously, $L_{\mathbf{a}}(\alpha_i)=L_{a+\mathbf{a}}(a+\alpha_i)$.
This implies that when $1\leq tr^e \leq r^s-1$,
we can choose $\mathbf{a}\in F_{r^s}^*=H_1$.
Then by Lemma \ref{lem9}, there exists an MDS self-dual code of length $n=tr^e(1+r^s+\dots+r^{s(m-1)})$. This completes the proof.
\end{proof}

\begin{remark}
Theorem \ref{th8} generalizes known results.
\begin{itemize}
\item When $s=1$, we have $t(1+r+\dots+r^{m-1})\mid (q-1)$ and $q\equiv1 ({\rm mod}~4)$.
At this point, Theorem \ref{th8} is a special case of Theorem $1 (i)$ in \cite{RefJ10}.
\item When $m=1$, the preceding result is exactly Theorem \ref{th1}.
\end{itemize}
\end{remark}

By Lemma \ref{lem9} and Theorem \ref{th8}, we have the following corollary.

\begin{corollary}\label{cor1}
Let $q=r^{s m_1 m_2 \cdots m_l}$ and $q\equiv 1({\rm mod}~4)$, where $r$ is an odd prime power, $l$ and $s$ are positive integers, and $m_1m_2\cdots m_l$ is an odd integer. Let $t$ be an even integer such that $t \mid(r-1)$ and $0<t<r-1$, and let $0\leq e\leq s-1$. Then there exists a $q$-ary MDS self-dual code of length $t r^e (1+r^s+\cdots+r^{s(m_1-1)}) (1+r^{s m_1}+\cdots+r^{s m_1(m_2-1)})\cdots (1+r^{s m_1 m_2 \cdots m_{l-1} }+\cdots+r^{s m_1 m_2 \cdots m_{l-1} (m_l-1)})$.
\end{corollary}

\begin{theorem}\label{th9}
Let $q=r^{s m}$, where $r$ is an odd prime power, $m$ is an odd integer and $s$ is a positive integr. Let $t$ be an odd integer such that $t\mid (r-1)$, and let $0\leq e\leq s-1$. If $\eta(-t)=1$, then there exists a $q$-ary self-dual MDS code of length $(t+1) r^e(1+r^s+\cdots+r^{s(m-1)})$.
\end{theorem}

\begin{proof}
By \cite[Theorem 3.4]{RefJ6} and the proof of Theorem \ref{th8},
there exists a self-dual $GRS_{\frac{(t+1)r^e}{2}}(\mathbf{a}, \mathbf{v})$ for some $\mathbf{a}=\{\alpha_1,\alpha_2,\dots,\alpha_{(t+1)r^e}\}\subseteq H_1$,
$(t+1)r^e$ is even, and $\mathbf{v}\in (F_q^*)^{(t+1)r^e}$. By Lemma \ref{lem9}, there exists an MDS self-dual code of length $n=(t+1)r^e(1+r^s+\cdots+r^{s(m-1)})$. This completes the proof.
\end{proof}

By Lemma \ref{lem9} and Theorem \ref{th9}, we have the following corollary.

\begin{corollary}\label{cor2}
Let $q=r^{s m_1 m_2 \cdots m_l}$, where $r$ is an odd prime power, $s$ and $l$ are positive integers, and $m_1m_2\cdots m_l$ is an odd integer. Let $t$ be an odd integer and $t\mid (r-1)$, and let $0\leq e\leq s-1$. If $\eta(-t)=1$, then there exists a $q$-ary MDS self-dual code of length $(t+1)r^e(1+r^s+\dots+r^{s(m_1-1)})(1+r^{sm_1}+\dots+r^{sm_1(m_2-1)})\dots(1+r^{sm_1m_2\dots m_{l-1}}+\dots+r^{sm_1m_2\dots m_{l-1}(m_l-1)})$.
\end{corollary}

\begin{lemma}\label{lem10}
Let $t$ be odd. Suppose $GRS_{\frac{t+1}{2}}(\mathbf{a},\mathbf{v},\infty)$ is self-dual for some $\mathbf{a}=\{\alpha_1,\alpha_2,\dots, \alpha_t\}\subseteq H_1$ and $\mathbf{v}\in( F_q^*)^{t}$. If $e_1$ is odd and $\eta(e_1)=1$, then there exists a $q$-ary MDS self-dual code of length $te_1+1$.
\end{lemma}

\begin{proof}
Define
\begin{equation}
W=\bigcup_{1\leq i \leq t}(\theta^{v(\alpha_i)}H_2).
\end{equation}
Write all these $n=te_1$ elements as a vector $\mathbf{b}=(\beta_1,\beta_2,\dots,\beta_n)$.
By Lemma \ref{lem5},
\begin{equation}
f_{\theta^{v(\alpha_i)}H_2}(x)=x^{e_1}-\theta^{v(a_i)e_1}~and~L_{\theta^{v(\alpha_i)} H_2}(x)=f'_{\theta^{v(\alpha_i)} H_2}(x)=e_1x^{e_1-1}.
\end{equation}
If $\beta_i\in \theta^{v(\alpha_k)}H_2$ for some $k$, then there exists an integer $0\leq u\leq e_1-1$ such that $\beta_i=\theta^{v(\alpha_k)+e_2u}$.
Thus,
\[\begin{split}
	L_{\mathbf{b}}(\beta_i)&=L_{\theta^{v(\alpha_k)} H_2}(\beta_i)\cdot \prod_{1\leq j\leq t, j\neq k}f_{\theta^{v(\alpha_j)} H_2}(\beta_i)\\
                           &=e_1\theta^{v(\alpha_k)(e_1-1)}\theta^{-e_2u}\prod_{1\leq j\leq t, j\neq k}(\theta^{(v(\alpha_k)+e_2u)e_1}-\theta^{v(\alpha_j)e_1})\\
                           &=e_1\theta^{v(\alpha_k)(e_1-1)}\theta^{-e_2u}L_{\mathbf{a}}(\alpha_k).\\
	\end{split}\]

By \cite[Lemma 2.2]{RefJ9}, there exists $\mathbf{v}=(v_1,v_2,\dots,v_t)\in (F_q^*)^n$ such that  $GRS_{\frac{t+1}{2}}(\mathbf{a},\mathbf{v},\infty)$ is self-dual if and only if  $\eta(-L_{\mathbf{a}}(\alpha_k))=1$. Note that $e_1$ is odd and $e_2$ is even, we have
$$\eta(-L_{\mathbf{b}}(\beta_i))=\eta(e_1\cdot (-L_{\mathbf{a}}(\alpha_k)))=\eta(e_1)=1.$$
By Lemma \ref{lem2},
there exists a $q$-ary MDS self-dual code of length $te_1+1$.
\end{proof}

\begin{remark}\label{n(f)=1}
In Lemma \ref{lem10},
We can change condition $\eta(e_1)=1$ to $q\equiv 1(${\rm mod} $\,4)$. When $q\equiv 1({\rm mod}\,4)$. Since $(\frac{q-1+1}{e_1})=(\frac{1}{e_1})=1$, by Quadratic Reciprocity, we have $(\frac{q}{e_1})(\frac{e_1}{q})=(-1)^{\frac{q-1}{2}\frac{e_1-1}{2}}=1$. Therefore, $\eta(e_1)=1$.
\end{remark}

\begin{theorem}\label{th10}
Let $q=r^{s m}$, where $r$ is an odd prime power, $m$ is an odd integer, and $s$ is a positive integer. Let $t$ be an odd integer such that $t\mid (r-1)$, and let $0\leq e \leq s-1$. If $\eta((-1)^{\frac{r^e+1}{2}}t)=1$, then there exists a $q$-ary MDS self-dual code of length $tr^e(1+r^s+\dots+r^{s(m-1)})+1$.
\end{theorem}

\begin{proof}
By \cite[Theorem 3.2]{RefJ6}  and the proof of Theorem \ref{th8},
there exists a self-dual $GRS_{\frac{tr^e+1}{2}}(\mathbf{a},\mathbf{v},\infty)$,
for some $\mathbf{a}=\{\alpha_{1,1},\dots,\alpha_{1,r^e},\dots,\alpha_{t,r^e}\}\subseteq H_1= F_{r^s}^*$,
and $v\in (F_q^*)^t$.
Note that $\eta(1+r^s+\dots +r^{s(m-1)})=1$, it then follows from Lemma \ref{lem10} that there exists a $q$-ary MDS self-dual code of length $tr^e(1+r^s+\dots +r^{s(m-1)})+1$. This completes the proof.
\end{proof}

By Lemma \ref{lem10} and Theorem \ref{th10}, we have the following result.

\begin{corollary}\label{cor3}
Let $q=r^{s m_1m_2 \cdots m_l}$, where $r$ is an odd prime power, $s$ and $l$ are positive integers, and $m_1m_2 \cdots m_l$ is odd. Let $t$ be an odd integer such that $t\mid(r-1)$, and let $0\leq e \leq s-1$. If $\eta((-1)^{\frac{r^e+1}{2}}t)=1$, then there exists a $q$-ary MDS self-dual code of length $t r^e(1+r^s+\dots+r^{s(m_1-1)})(1+r^{s m_1}+\cdots+r^{s m_1(m_2-1)})\cdots(1+r^{s m_1m_2\cdots m_{l-1}}+\cdots+r^{s m_1 m_2\cdots m_{l-1}(m_l-1)})+1$.

\end{corollary}

\begin{theorem}\label{th11}
Let $q=r^{s m}$, where $r$ is an odd prime power, $m$ is an odd integer and $s$ is a positive integer. Let $t$ be an even integer such that $t\mid(r-1)$ and $1\leq t< r-1$, and let $0\leq e \leq s-1$. If $\eta(t)=\eta(-1)=1$, or $\eta(-t)=1$ and $e$ is even, then there exists a $q$-ary MDS self-dual code of length $(t+1)r^e(1+r^s+\dots +r^{s(m-1)})+1$.
\end{theorem}

\begin{proof}
By Theorem \ref{th4} and the proof of Theorem \ref{th8},
there exists a self-dual $GRS_{\frac{(t+1)r^e+1}{2}}(\mathbf{a}, \mathbf{v},\infty)$,
for some $\mathbf{a}=\{\alpha_1,\alpha_2,\dots,\alpha_{(t+1)r^e}\}\subseteq H_1$,
$(t+1)r^e$ is odd, and $\mathbf{v}\in (F_q^*)^{(t+1)r^e}$.
Then by Lemma \ref{lem9},
there exists an MDS self-dual code of length $n=(t+1)r^e(1+r^s+\dots+r^{s(m-1)})+1$. This completes the proof.
\end{proof}

\begin{remark}
Theorem \ref{th11} generalizes known results.
When $m=1$, the preceding result is exactly Theorem \ref{th4}.
\end{remark}

Using Lemma \ref{lem10} and Theorem \ref{th11}, we have the following result.

\begin{corollary}\label{cor4}
Let $q=r^{s m_1 m_2 \cdots m_l}$, where $r$ is an odd prime power, $s$ and $l$ are positive integers, and $m_1m_2 \cdots m_l$ is odd. Let $t$ be an even integer such that $t\mid (r-1)$ and $1\leq t< r-1$, and let $0\leq e \leq s-1$. Then there exists a $q$-ary MDS self-dual code of length $(t+1)r^e(1+r^s+\cdots+r^{s(m_1-1)})(1+r^{s m_1}+\cdots+r^{s m_1(m_2-1)})\cdots (1+r^{s m_1 m_2\cdots m_{l-1}}+\cdots+r^{s m_1m_2\cdots m_{l-1}(m_l-1)})+1$.
\end{corollary}

When $q-1$ has two decompositions, we have the following result.

\begin{lemma}\label{lem two dec}
Let $q$ be a prime power. Suppose $q-1=e_1f_1=e_2f_2$, $F_q^*=\langle \theta\rangle$, $H=\langle \alpha \rangle$ and $M=\langle \beta\rangle$, where $\alpha=\theta^{e_1}$ and $\beta=\theta^{e_2}$. Let $\{\beta^{i_1},\beta^{i_2},\dots,\beta^{i_t}\}$ be a subset of $M$, where $\{i_1,i_2,\dots,i_t\}$ are distinct modulo $f_2$. Then $\beta^{i_\lambda}H(1\leq \lambda \leq t)$ are distinct cosets of $H$ in $F_q^*$ if and only if $i_1, i_2,\dots, i_t$ are distinct modulo $\frac{f_2}{\gcd(f_2,f_1)}$.
\end{lemma}

\begin{proof}
For $i, j\in Z_{f_2}$,
\[\begin{split}
	    \beta^iH=\beta^jH  &\Leftrightarrow \beta^{i-j}= \theta^{e_2(i-j)}\in H=\langle \theta^{e_1}\rangle\\
                           &\Leftrightarrow e_2(i-j)\equiv 0~({\rm mod} \,e_1)\\
                           &\Leftrightarrow \frac{e_1}{\gcd(e_1,e_2)}\mid i-j.\\
	\end{split}\]
Since $q-1=e_1f_1=e_2f_2$, we have
\[\begin{split}
	    \frac{e_1}{\gcd(e_1,e_2)}&=(\frac{e_2f_2}{f_1})\frac{1}{\gcd(\frac{e_2f_2}{f_1},e_2)}\\
                           &=\frac{e_2f_2}{\gcd(e_2f_2,e_2f_1)}\\
                           &=\frac{f_2}{\gcd(f_2,f_1)}.\\
	\end{split}\]
Therefore, $\beta^iH=\beta^jH$ if and only if $i\equiv j ({\rm mod}~\frac{f_2}{\gcd(f_2,f_1)})$. This completes the proof.
\end{proof}

\begin{lemma}\label{lem two tf1}
Let $t$ be even. Suppose $GRS_{\frac{t}{2}}(\mathbf{a},\mathbf{v})$ is self-dual for some $\mathbf{a}=\{\beta^{i_1f_1},\beta^{i_2f_1},\dots, \beta^{i_tf_1} \}\subseteq \langle \theta^{e_2f_1} \rangle$ and $\mathbf{v}\in (F_q^*)^{t}$. If $e_1$ and $e_2(f_1-1)$ are even, then there exists a $q$-ary MDS self-dual code of length $tf_1$.
\end{lemma}

\begin{proof}
Obviously, $\{i_1,i_2,\dots,i_t\}$ is a subset of $Z_{f_2}$ such that
$i_1,i_2,\dots,i_t$ are distinct modulo $\frac{f_2}{\gcd(f_2,f_1)}$. By Lemma \ref{lem two dec},
$\beta^{i_\lambda}H(1\leq \lambda\leq t)$ are distinct.
Then
\begin{equation}
S=\bigcup_{\lambda=1}^t\beta^{i_\lambda}H
\end{equation}
is a union of $t$ cosets of $H$ in $F_q^*$ and $|S|=t f_1$. For each $\gamma=\beta^{i_\lambda}\alpha^j\in S$,
by Lemma \ref{lem xin}, we have
\[\begin{split}
	        L_{S}(\gamma)  &=f'_S(\gamma)=g'(\gamma^{f_1})f_1\gamma^{f_1-1}\\
                           &=f_1\beta^{i_\lambda (f_1-1)}\theta^{je_1(f_1-1)}L_{\mathbf{a}}(\beta^{i_\lambda f_1})\\
                           &=f_1\theta^{e_2 i_\lambda (f_1-1)}\theta^{-je_1}L_{\mathbf{a}}(\beta^{i_\lambda f_1}),\\
	\end{split}\]
where $\mathbf{a}=\{\beta^{i_1f_1},\beta^{i_2f_1},\dots, \beta^{i_tf_1} \}\subseteq \langle\theta^{e_2f_1}\rangle$.
Note that if $e_1$ and $e_2(f_1-1)$ are even,
then $\eta(L_S(\gamma))$ are same for all $\gamma\in S$.
By Remark \ref{rem1}, there exists a $q$-ary MDS self-dual code of length $t f_1$. This completes the proof.
\end{proof}

\begin{remark}
Lemma \ref{lem two tf1} compared with Lemma \ref{lem9},
we expanded the scope of $f_1$,
so that $f_1$ can take an even number.
\end{remark}

\begin{lemma}\label{lem two tf11}
Let $t$ be even. Suppose $GRS_{\frac{t}{2}}(\mathbf{a},\mathbf{v})$ is self-dual for some $\mathbf{a}=\{\beta^{i_1f_1^2},\beta^{i_2f_1^2},\dots, \beta^{i_tf_1^2} \}\subseteq \langle\theta^{e_2f_1^2}\rangle$ and $\mathbf{v}\in (F_q^*)^{t}$. If $e_1$ is even, then there exists a $q$-ary MDS self-dual code of length $tf_1$.
\end{lemma}

\begin{proof}
Since $\mathbf{a}=\{\beta^{i_1f_1^2},\beta^{i_2f_1^2},\dots, \beta^{i_tf_1^2} \}\subseteq \langle\theta^{e_2f_1^2}\rangle\subseteq \langle\theta^{e_2f_1}\rangle$, we have
 $$\beta^{i_1f_1^2},\beta^{i_2f_1^2},\dots, \beta^{i_tf_1^2}$$ are different elements in $\langle\theta^{e_2f_1}\rangle$. It follows that $i_1f_1,i_2f_1,\cdots,i_tf_1$ are distinct modulo $\frac{f_2}{\gcd(f_2,f_1)}$. By Lemma \ref{lem two dec}, $\beta^{i_\lambda f_1}H(1\leq \lambda\leq t)$ are distinct.
Then
\begin{equation}
S=\bigcup_{\lambda=1}^t\beta^{i_\lambda f_1}H
\end{equation}
is a union of $t$ cosets of $H$ in $F_q^*$ and $|S|=t f_1$. For each $\gamma=\beta^{i_\lambda f_1}\alpha^j\in S$, by Lemma \ref{lem xin}, we have
\[\begin{split}
	        L_{S}(\gamma)  &=f'_S(\gamma)=g'(\gamma^{f_1})f_1\gamma^{f_1-1}\\
                           &=f_1\beta^{i_\lambda f_1(f_1-1)}\theta^{-je_1}L_{\mathbf{a}}(\beta^{i_\lambda f_1^2}),\\
	\end{split}\]
where $\mathbf{a}=\{\beta^{i_1f_1^2},\beta^{i_2f_1^2},\dots, \beta^{i_tf_1^2} \}\subseteq \langle\theta^{e_2f_1^2}\rangle$.
Note that if $e_1$ is even, $\eta(L_S(\gamma))$ are same for all $\gamma\in S$. By Remark \ref{rem1},
there exists a $q$-ary MDS self-dual code of length $t f_1$. This completes the proof.
\end{proof}

Using the above lemmas, we have the following result, which can be regarded as the generalization of some results in \cite{RefJ20}.

\begin{theorem}\label{th tf tf+2}
Let $q=r^2$, where $r$ is an odd prime power. Let $q-1=e f$, where $e, f$ are positive integers. Let $s$ be a positive integer such that $s\mid f$ and $s\mid(r-1)$, and let $D=\frac{s(r+1)}{\gcd(s(r+1),f)}$. Let $1\leq t\leq D$ and $tf$ is even. The following hold.
\begin{itemize}
\item[(1)] If $e$ and $\frac{r-1+ft}{s}$ are even, then there exists a $q$-ary MDS self-dual code of length $t f$.
\item[(2)] When $1\leq t\leq D-1$, if $\frac{f}{s}$ and $\frac{1}{2}(t-1)(r+1)$ are even, or, $\frac{f}{s}$ is odd and $t$ is even, then there exists a $q$-ary MDS self-dual code of length $t f+2$.
\item[(3)] When $t=D$, if $\frac{ft}{s}$ and $\frac{t-1}{2}(r+1-\frac{ft}{s})$ are even, then there exists a $q$-ary MDS self-dual code of length $tf+2$.
\end{itemize}
\end{theorem}

\begin{proof}
(1) Let $F_q^*=\langle\theta\rangle$, $\beta=\theta^{\frac{r-1}{s}}$, $\alpha=\theta^e$ and $H=\langle\alpha\rangle$. By Lemma \ref{lem two dec}, we can prove that $S=\bigcup_{\lambda=1}^t\beta^{i_\lambda} H $ is a disjoint union of $t$ coests of $H$ in $F_q^*$ and  $|S| =t f$. For $\gamma=\beta^{i_\mu}\alpha^j\in S$, by the proof of Lemma \ref{lem two tf1}, we have
\begin{equation}\label{eq3}
 L_{S}(\gamma)=f\beta^{i_\mu(f-1)}\theta^{-je}L_{\mathbf{a}}(\beta^{i_\mu f}),
\end{equation}
where $\mathbf{a}=\{\beta^{i_\mu f}: 1\leq \mu\leq t\}$. It follows from $s\mid f $ that
\begin{equation}
(\beta^{i_\mu f})^r=\theta^{\frac{f}{s}i_\mu r(r-1)}=\theta^{\frac{f}{s}i_\mu (1-r)}=\beta^{-i_\mu f}.
\end{equation}
Therefore,
\[\begin{split}
	        (L_{\mathbf{a}}(\beta^{i_\mu f}))^r &= \prod_{\lambda=1 ,\lambda\neq \mu}^{t}(\beta^{-i_\mu f}-\beta^{-i_\lambda f}) \\
                                              &= \prod_{\lambda=1 ,\lambda\neq \mu}^{t}\frac{\beta^{i_\lambda f}-\beta^{i_\mu f}}{\beta^{(i_\mu+i_\lambda)f}}\\
                                              &= (-1)^{t-1}\beta^{-i_\mu f(t-2)-fI} L_{\mathbf{a}}(\beta^{i_\mu f}),\\
	\end{split}\]
where $I=\sum_{\lambda=1}^t i_\lambda$. It then follows that
\begin{equation}
(L_{\mathbf{a}}(\beta^{i_\mu f}))^{r-1}=\theta^{\frac{1}{2}(t-1)(r^2-1)-f\frac{r-1}{s}[i_\mu (t-2)+I]}
\end{equation}
and
\begin{equation}
L_{\mathbf{a}}(\beta^{i_\mu f})=\theta^{\frac{1}{2}(t-1)(r+1)-\frac{f}{s}[i_\mu (t-2)+I]+(r+1)c},
\end{equation}
for some $c\in Z$. Therefore,
\[\begin{split}
	        L_{S}(\gamma)&=f \theta^{i_\mu(f-1)\frac{r-1}{s}-je+\frac{1}{2}(t-1)(r+1)-\frac{f}{s}[i_\mu (t-2)+I]+(r+1)c}\\
                         &=f\theta^{i_\mu[\frac{(f-1)(r-1)-f(t-2)}{s}]}\theta^{-je}\theta^{\frac{1}{2}(t-1)(r+1)-\frac{fI}{s}+(r+1)c}.\\
 \end{split}\]
Consequently,
\[\begin{split}
	        \eta(L_{S}(\gamma))&=(-1)^{i_\mu[\frac{(r-1)+ft}{s}]}(-1)^{-je}(-1)^{\frac{1}{2}(t-1)(r+1)-\frac{fI}{s}}.\\
 \end{split}\]
If $e$ and $\frac{(r-1)+ft}{s}$ are even,
then $\eta(L_{S}(\gamma))$ are same for all $\gamma\in S$.
By Remark \ref{rem1},
there exists a $q$-ary MDS self-dual code of length $tf$.

(2) Let $\tilde{S}=S\cup\{0\}$ and $|\tilde{S}|=tf+1$. For $\gamma=\beta^{i_\mu}\alpha^j\in S$,
$L_{\tilde{S}}(\gamma)=L_S(\gamma)\gamma$. By (\ref{eq3}), we have
\[\begin{split}
	       L_{\tilde{S}}(\gamma)&=f\beta^{i_\mu f}L_{\mathbf{a}}(\beta^{i_\mu f})\\
                                &=f\theta^{i_\mu[\frac{f(r-1)-f(t-2)}{s}]}\theta^{\frac{1}{2}(t-1)(r+1)-\frac{fI}{s}+(r+1)c},\\
 \end{split}\]
for some integer $c\in Z$. Hence,
\[\begin{split}
	     \eta(-L_{\tilde{S}}(\gamma))&=(-1)^{i_\mu[\frac{f(r-1-t)}{s}]+\frac{1}{2}(t-1)(r+1)-\frac{fI}{s}} \\
                                &=(-1)^{i_\mu(\frac{ft}{s})}(-1)^{\frac{1}{2}(t-1)(r+1)-\frac{fI}{s}}.\\
 \end{split}\]
 If $\frac{ft}{s}$ and $\frac{1}{2}(t-1)(r+1)-\frac{fI}{s}$ are even, we have $\eta(-L_{\tilde{S}}(\gamma))=1$. For $1\leq t\leq D-1$, there two cases:
\begin{itemize}
\item Case 1. If $\frac{f}{s}$ and $\frac{1}{2}(t-1)(r+1)$ are even, we have $\eta(-L_{\tilde{S}}(\gamma))=1$.
\item Case 2. If $\frac{f}{s}$ is odd and $t$ is even, we can always choose suitable $i_1,i_2,\dots,i_t$ such that $\frac{r+1}{2}+I$ is even. Hence, $\eta(-L_{\tilde{S}}(\gamma))=1$.
\end{itemize}
For $0\in \tilde{S}$,
\[\begin{split}
	       L_{\tilde{S}}(0)&=(-1)^{tf}\prod_{\lambda=1}^{t}\prod_{j=0}^{f-1}(\beta^{i_\lambda}\alpha^j)\\
                                &=(-1)^{tf}\beta^{fI}\alpha^{\frac{tf(f-1)}{2}}\\
                                &=(-1)^{tf}\theta^{\frac{(r-1)fI}{s}}\theta^{\frac{etf(f-1)}{2}}.\\
 \end{split}\]
 Since $ef=q-1=r^2-1\equiv 0({\rm mod}~4)$,
 $$\eta(-L_{\tilde{S}}(0))=(-1)^{\frac{(r-1)fI}{s}}=1.$$
 By Lemma \ref{lem2}, there exists a $q$-ary MDS self-dual code of length $t f+2$.

(3) Let $\gamma=\beta^{i_\mu}\alpha^j\in S$, For $t=D$, note that $I=\sum_{\lambda=1}^t i_\lambda=\frac{t(t-1)}{2}$.
If $\frac{ft}{s}$ and $\frac{t-1}{2}(r+1-\frac{ft}{s})$ are even, then $\eta(-L_{\tilde{S}}(\gamma))=1$.
For $0\in \tilde{S}$, by the proof of (2),
 $$\eta(-L_{\tilde{S}}(0))=(-1)^{\frac{(r-1)fI}{s}}=1.$$
 By Lemma \ref{lem2}, there exists a $q$-ary MDS self-dual code of length $t f+2$. This completes the proof.
\end{proof}

Similarly, we have the following result, which can be regarded as the generalization of some results in \cite{RefJ20} and \cite{RefJ15}.

\begin{theorem}\label{th tf tf+1 tf+2}
Let $q=r^2$, where $r$ is an odd prime power. Let $q-1=e f$, where $e$ and $f$ are positive integers. Let $s$ be a positive integer such that $s\mid f$ and $s\mid (r+1)$. Let $D=\frac{s(r-1)}{\gcd(s(r-1),f)}$ and $1\leq t\leq D$. If $tf$ is odd, then there exists a $q$-ary MDS self-dual code of length $tf+1$.
\end{theorem}

\begin{proof}
Let $F_q^*=\langle\theta\rangle$, $\beta=\theta^{\frac{r+1}{s}}$, $\alpha=\theta^e$ and $H=\langle\alpha\rangle$. By Lemma \ref{lem two dec},
we can prove that $S=\bigcup_{\lambda=1}^t\beta^{i_\lambda} H $ is a disjoint union of $t$ coests of $H$ in $F_q^*$ and $|S|=t f$. For $\gamma=\beta^{i_\mu}\alpha^j\in S$, by the proof of Lemma \ref{lem two tf1}, we have
\begin{equation}\label{eq4}
 L_{S}(\gamma)=f\beta^{i_\mu(f-1)}\theta^{-je}L_{\mathbf{a}}(\beta^{i_\mu f}),
\end{equation}
where $\mathbf{a}=\{\beta^{i_\mu f}: 1\leq \mu\leq t\}$. It follows from $s\mid f $ that
\begin{equation}
(\beta^{i_\mu f})^r=\theta^{\frac{f}{s}i_\mu r(r+1)}=\theta^{\frac{f}{s}i_\mu (r+1)}=\beta^{i_\mu f}.
\end{equation}
Hence, $(L_{\mathbf{a}}(\beta^{i_\mu f}))^r=L_{\mathbf{a}}(\beta^{i_\mu f})$. It then follows that $L_{\mathbf{a}}(\beta^{i_\mu f})\in F_r$ and
$$\eta(L_{\mathbf{a}}(\beta^{i_\mu f}))=1.$$
Note that $e$ and $f-1$ are even,
so $\eta(-L_{S}(\gamma))=1$ for all $\gamma\in S$.
By Lemma \ref{lem2},
there exists a $q$-ary MDS self-dual code of length $tf+1$.
\end{proof}

\begin{lemma}\label{lem lager q}
Let $q$ be an odd prime power. For any given $n$, if $q> (t+(t^2+(n-1)2^{n-2})^{\frac{1}{2}})^2$, where $t=(n-3)2^{n-3}+\frac{1}{2}$,
then there exists a subset $\mathbf{a}=\{\alpha_1,\alpha_2,\dots,\alpha_n\}$ of $F_q$ such that $\alpha_j-\alpha_i$ are nonzero square elements for all $1\leq i<j\leq n$.
\end{lemma}

\begin{proof}
We prove it by induction on $n$.
For $n=2$, we can let $\mathbf{a}=\{0,1\}$.
Suppose that there exists a subset $\mathbf{a}=\{\alpha_1,\alpha_2,\dots,\alpha_{n-1}\}$ of $F_q$ of size $n-1$ such that
$\alpha_j-\alpha_i$ are nonzero square elements for all $1\leq i<j\leq n-1$.
Let $N$ denote the number of elements $\beta$ of $F_q$ such that $\eta(\beta-\alpha_i)=1$ for all $i=1,2,\dots,n-1$.
Then by [\cite{RefJ19}, Exercise 5.64], one has
 \begin{equation}
\left| N-\frac{q}{2^{n-1}}\right| \leq \left(\frac{n-3}{2}+\frac{1}{2^{n-1}}\right)\sqrt{q}+\frac{n-1}{2}.
\end{equation}
Thus we have
\begin{equation}
 N\geq \frac{q}{2^{n-1}}-\left(\frac{n-3}{2}+\frac{1}{2^{n-1}}\right)\sqrt{q}-\frac{n-1}{2}.
\end{equation}
Let $y(x)=\frac{x^2}{2^{n-1}}-(\frac{n-3}{2}+\frac{1}{2^{n-1}})x-\frac{n-1}{2}$,
when $x> t+(t^2+(n-1)2^{n-2})^{\frac{1}{2}}$, we have $y(x)> 0$. By our condition on $n$ and $q$, this implies that exists an element $\alpha_n$ such that $\alpha_n-\alpha_i$ are nonzero square elements of $F_q$, for all $i=1,2,\dots,n-1$. This completes the proof.
\end{proof}

\begin{theorem}\label{th lager q}
Let $q\equiv 1({\rm mod}~4)$ and let $n$ be an even integer such that $q>(t+(t^2+(n-1)2^{n-2})^{\frac{1}{2}})^2$, where $t=(n-3)2^{n-3}+\frac{1}{2}$. Then there exists a $q$-ary self-dual MDS code of length $n$.
\end{theorem}

\begin{proof}
By Lemma \ref{lem lager q}, there exists a subset $\mathbf{a}=\{\alpha_1,\alpha_2,\dots,\alpha_n\}$
such that $\alpha_i-\alpha_j$ are square elements for all $1\leq i<j\leq n$. Since $q\equiv 1({\rm mod}~4)$, $\eta(-1)=1$. Therefore, $\beta-\gamma$ is a nonzero square for any two distinct elements $\beta,\gamma\in \mathbf{a}$. It follows that $\eta(L_{\mathbf{a}}(\alpha_i))=1$ for all $1\leq i\leq n$. Then there exists a $q$-ary self-dual MDS code with length $n$. This completes the proof.
\end{proof}

\begin{remark}
Theorem \ref{th lager q} generalizes the previous conclusion,
compared with Theorem 3.2 (ii) in \cite{RefJ7}, the range of $q$ in Theorem \ref{th lager q} is wider.
\end{remark}

\section{Conclusions}\label{sec-conclusion}
In this paper we construct some new classes of $q$-ary MDS self-dual codes of different lengths by (extended) generalized Reed-Solomon codes (see Theorems \ref{th1}, \ref{th2}, \ref{th3}, \ref{th4}, \ref{th8}, \ref{th9}, \ref{th10}, \ref{th11}, \ref{th tf tf+2}, \ref{th tf tf+1 tf+2}, \ref{th lager q}). We generalize some previous works on the existence of MDS self-dual codes, and some known results can be considered as special cases of this construction.

%Finally, combining Lemmas \ref{lem10}, \ref{lem two tf1} and \ref{lem two tf11}, we hope more MDS self-dual codes with new parameters can be obtained.

\section*{Acknowledgments}
%The authors thank the editor and anonymous referees for their work to improve the readability of this paper.
This research was supported by the National Natural Science Foundation of China (No.U21A20428 and 12171134).

%\section*{References}


\begin{thebibliography}{}

    \bibitem{RefJ1}R. Cramer, V. Daza, I. Gracia, J.J. Urroz, G. Leander, J. Marti-Farre, C. Padro, On codes, matroids and secure multi-party computation from linear secret sharing schemes, IEEE Trans. Inf. Theory 54 (6) (2008) 2647-2657.
	
	\bibitem{RefJ2} S.H. Dau, W. Song, Z. Dong, C. Yuen, Balanced sparsest generator matrices for MDS codes, in: Proc. IEEE Inter. Symp. Inform. Theory, 2013, pp. 1889-1893.
	
	\bibitem{RefJ3} M. Grassl, T.A. Gulliver, On self-dual MDS codes, in: Proc. IEEE Inter. Symp. Inform. Theory, 2008, pp. 1955-1957.
	
	\bibitem{RefJ4}W. Fang, F.W Fu, New constructions of MDS Euclidean self-dual codes from GRS codes and extended GRS codes, IEEE Trans. Inf. Theory 65 (9) (2019) 5574-5579.
	
	\bibitem{RefJ5} J.L. Massey, Some applications of coding theory in cryptography, in: Proc. 4th IMA Conf. Cryptogr. Coding, 1995, pp. 33-47.
	
	\bibitem{RefJ6} K. Lebed, H. Liu, Some new constructions of MDS self-dual codes over finite fields, Finite Fields Appl. 77 (2022) 101-934.
	
	\bibitem{RefJ7}L. Jin, C. Xing, New MDS self-dual codes from generalized Reed-Solomon codes, IEEE Trans. Inf. Theory 63 (3) (2017) 1434-1438.
	
	\bibitem{RefJ8} T.A. Gulliver, J.L. Kim, Y. lee, New MDS or Near-MDS self-dual codes, IEEE Trans. Inf. Theory, 2008, pp. 4354-4360.
	
	\bibitem{RefJ9} D. Xie, X. Fang, J. Luo, Construction of long MDS self-dual codes from short codes, Finite Field Appl. 72 (2021) 101-813.
	
	\bibitem{RefJ10} H. Yan, A note on the construction of MDS self-dual codes, Cryptogr. Commun. 11 (2) (2019) 259-268.

	\bibitem{RefJ11}S. Georgiou, C. Koukouvinos, MDS self-dual codes over large prime fields, Finite Fields Appl. 8 (4) (2002) 455-470.

	\bibitem{RefJ12} J.L. Kim, Y. Lee, Euclidean and Hermitian self-dual MDS codes over large finite fields, J. Comb. Theory Ser. A 105 (1) (2004) 79-95.
	
	\bibitem{RefJ13} A. Zhang, K. Feng, A unified approach to construct MDS self-dual codes via Reed-Solomon codes, IEEE Trans. Inf. Theory, 66 (2020) 3650-3655.
	
	\bibitem{RefJ14}K. Lebed, H. Liu, J. Luo, Construction of MDS self-dual codes over finite fields, Finite Fields Appl. 59 (2019) 199-207.
	
    \bibitem{RefJ15} X. Fang, K. Lebed, H. Liu, J. Luo, New MDS self-dual codes over finite fields of odd characteristic, Des. Codes Cryptogr. 88 (2) (2020) 1127-1138.
	
	\bibitem{RefJ16} W. Fang, F.W. Fu, Construction of MDS Euclidean self-dual codes via two subsets, IEEE Trans. Inf. Theory 67 (8) (2021) 5005-5015.
	
	\bibitem{RefJ17} H. Tong, X. Wang, New MDS Euclidean and Hermitian self-dual codes over finite fields, Adv. Pure Math. 7 (5) (2017) 325-333.

    \bibitem{RefJ18} X. Fang, M. Liu, J. Luo, New MDS Euclidean self-orthogonal codes, IEEE Trans. Inf.Theory, 75 (2021) 130-137.

    \bibitem{RefJ19} R. Lidl and H. Niederreiter, Finite Fields. Cambridge, U.K.: Cambridge Univ. Press, 1997.

    \bibitem{RefJ20} W. Fang, J. Zhang, S.T. Xia,, F.W. Fu, New constructions of self-dual generalized Reed-Solomon codes. Cryptogr. Commun. 14 (2022) 677-690.

\end{thebibliography}
\end{document}